\documentclass{article}

\usepackage{amsthm}
\usepackage{amssymb}
\usepackage{mathtools}

\newcommand{\eps}{\varepsilon}
\newcommand{\dist}{\mathrm{dist}}

\newcommand{\C}{\mathrm{C}\mskip 0.4mu}
\newcommand{\I}{\mathrm{I}\mskip 0.4mu}
\newcommand{\cnd}{\mskip 1mu | \mskip 1mu}

\newcommand\eqdef{\stackrel{\mathclap{\normalfont\mbox{\tiny def}}}{=}}

\newtheorem{theorem}{Theorem}
\newtheorem*{theorem*}{Main Result}
\newtheorem*{theorem-bis*}{Theorem $\mathbf{\ref{th:main}'}$}
\newtheorem{lemma}{Lemma}
\newtheorem*{claim}{Claim}
\theoremstyle{remark}
\newtheorem{definition}{Definition}
\newtheorem*{remark*}{Remark}

\usepackage{graphicx}
\usepackage{tikz}
\usepackage{caption,subcaption}

\begin{document}

\title{Clustering with Respect to the Information Distance}
\author{Andrei Romashchenko}

\maketitle

\begin{abstract}
We discuss the notion of a dense cluster with respect to the information distance and prove that all such clusters have an extractable core that represents the mutual information shared by the objects in the cluster.
\end{abstract}

\section{Introduction}

In the seminal paper \cite{InformationDistance1998},  Bennett et al. introduced the notion of \emph{information distance} based on Kolmogorov complexity. Loosely speaking, the distance between two individual finite objects is defined as the length of the shortest program that can translate these objects to each other. A surprising result proven in \cite{InformationDistance1998}  claims that the length of such a program (that performs the translation of the given objects to each other, in both directions) is substantially equal to the maximum of the lengths of two separate programs translating the objects to each other.  The optimal lengths of the programs depend, of course, on the choice of the programming language. Such a choice  may look arbitrary. However, the  framework of Kolmogorov complexity (see \cite{kolmogorov1965}) provides us with an \emph{optimal} programming language where the required programs have the minimum (up to an additive constant) possible length.

The notion of information distance attracted the attention of theoretical computer scientists and also inspired many experimental works, where various practical approximations of the information distance were computed for real-world data. This technique was typically used to reveal clusters and classify data of some specific type (texts, music recordings, genetic codes, and so on). In such experiments, the revealed clusters (groups of objects with small pairwise information distances) consist  of the data having something in common: texts written in the same language, music pieces of the same genre, genetic information of closely related species, and so on (see, e.g, \cite{text-clusters,music-clusters,clusters-compression}). 
In the studied practical  examples, it can be usually observed that the revealed dense clusters  have some core (e.g., the common language vocabulary, specific characteristics of a particular music genre, the genetic information of the common predecessor of several biological species), even if providing  an explicit description of such a core is not immediate.

In this paper, we study a similar phenomenon in a purely theoretical setting. We show that the \emph{only} reason why a large set of objects may form a dense cluster with respect to the information distance is that these objects share some \emph{common information} in the sense of G\'acs and K\"orner \cite{gacs-korner}. In other words, there must exist a \emph{core}  that is simple conditional on each object of the cluster and that represents in some sense the mutual information shared by all these objects.
Such results were used (more or less explicitly)  as technical tools in \cite{romashchenko-ccc-2003,romashchenko-ppi-2003,muchnik-romashchenko-ppi}.  We believe that the observed phenomenon is interesting in its own right, and we want to draw more attention to this issue. In this paper, we propose a self-contained explanation of the mentioned result using a simplified definition of a cluster proposed by Alexander Shen
(personal communication, June 10, 2021), see Definition~\ref{def:cluster} below.
 The proofs of the main results follow the arguments suggested in \cite{romashchenko-ccc-2003,romashchenko-ppi-2003}.

\paragraph{Notation and standard properties of Kolmogorov complexity}  In what follows we use the standard notation $\C(x)$ for the plain Kolmogorov complexity of a string $x$ and $\C(x \cnd y)$  for the plain Kolmogorov complexity of a string $x$  conditional on a string $y$ (see, e.g.,  \cite{kolmogorov-textbook} or \cite{svu-textbook}).  In this paper we do not discuss prefix-free complexity, monotone complexity,  or any other subtler versions of algorithmic complexity, so we may use for $\C(x)$  and $\C(x \cnd y)$  the term \emph{Kolmogorov complexity} without risk of ambiguity.

As usual, we assume to be fixed an \emph{encoding} of tuples, i.e.,  a computable bijection between  binary strings and tuples (finite ordered lists) of strings. Thus, every tuple of strings 
is associated with its \emph{code}, which is an individual binary string.  Keeping this in mind, we assume that  $\C(x,y)$ denotes Kolmogorov complexity of the code of the pair $\langle x,y\rangle $;  $\C(x,y,z)$ denotes Kolmogorov complexity of the code of the triple $\langle x,y,z\rangle$, and so on. 

Observe that any two encodings of this type  (computable bijections between tuples and individual strings) are equivalent: there are translation algorithms converting a code of a tuple in one encoding system into the code of the same tuple in the other encoding systems,  and the other way around.
This means that the choice of the encoding is quite arbitrary and affects the value of Kolmogorov complexity by at most  an additive constant, see \cite[Section~2.1]{svu-textbook}.

The classical Kolmogorov--Levin theorem (the \emph{chain rule}), \cite{kolmogorov1968,zvonkin-levin}, establishes the relation between Kolmogorov complexity of a pair and conditional Kolmogorov complexity,
\[
\C(x,y) = \C(x) + \C(y \cnd x)  + O(\log (\C(x)+ \C(y))).
\]
The mutual information of two strings and the conditional mutual information are defined as  $\I(x:y) \eqdef \C(y) - \C(y \cnd x)$ and $\I(x:y \cnd z) \eqdef  \C(y \cnd z) - \C(y \cnd x,z)$ respectively. 
From the Kolmogorov--Levin theorem it follows that mutual information is symmetric up to a logarithmic additive term:
\[
\I(x:y) = \C(x) + \C(y) - \C(x,y) + O(\log n) = \I(y:x) + O(\log n) 
\]
and 
\[
\I(x:y \cnd z) = \C(x,z)  +\C(y,z) - \C(x,y,z)  - \C(z)  + O(\log m) = \I(y:x \cnd z) + O(\log m), 
\]
where $n = \C(x) +\C(y)$ and $m = \C(x) +\C(y) + \C(z)$.

\section{Clusters with respect to the information distance}\label{s:2}

The \emph{information distance between} $x$ and $y$ can be  defined as 
\[
\dist(x, y) = \max\{ \C(x \cnd y), \C(y \cnd x) \}.
\]
(It is not  a distance in the proper sense since the triangle inequality is true only up to an additive logarithmic term.) Information distance measures the amount of information needed to obtain one of the strings given another one. Speaking informally, this value can be understood as a measure of  ``similarity'' (or rather ``non-similarity'') between strings: if $x,y,z$ are three strings of the same length, and $\dist(x, y)$ is much less than $\dist(x, z)$, we can say that $x$ is more ``similar'' to $y$ than to $z$.

We  can consider clusters defined in the sense of this information distance, i.e.,  large sets of strings with small diameters. Roughly speaking, a cluster is a set of strings of cardinality at least $2^m$ and diameter at most $m$. As  is usual in the theory of Kolmogorov complexity, we should admit a minor imprecision (say, logarithmic in $m$) of the parameters. The formal definition of a cluster involves two parameters,  the diameter and the logsize (logarithm of the cardinality):
\begin{definition}\label{def:cluster}
We say that a set of strings $S$ is an $(m,\ell)$-\emph{cluster}  if  for all $x_1,x_2\in S$ we have $\dist(x_1,x_2)\le m $ and $\# S \ge 2^{\ell}$. 
An equivalent wording: 
we  can say that this $S$ is a cluster with parameters $(m,\ell)$.
The minimal suitable value of $m$ is called the cluster's \emph{diameter} and the maximal suitable integer number $\ell$ is called  the cluster's \emph{logsize}.
\end{definition}
 The ``density'' of a cluster can be measured by the difference between the diameter and the logsize: the closer they are to each other, the denser is the cluster. 
 We usually deal with clusters where this  difference is bounded by  $O(\log m)$. 

First of all, do the dense clusters exist? The answer to this question is yes, we can find $(m,\ell)$-clusters with only a logarithmic gap between $m$ and $\ell$. Indeed, for every string $z$ we may consider a ``canonical'' cluster or a \emph{daisy} that consists of strings $x$ such that
$ \C(z \cnd x) \approx 0 $ and $\C(x \cnd z) \lesssim m$.
To make these approximate equality and inequality more specific, we fix a parameter $d$ and define a daisy with imprecision $d$ as follows.
\begin{definition}
An \emph{$(m,d)$-daisy with a core} $z$ is the set of all $x$ such that
\[ 
\C(z \cnd x) \le d \text{ and } C(x \cnd z) \le m + d.
\]
\end{definition}
Observe that the definition of a daisy involves two integer parameters, but they do not play the same role as the parameters  in the general definition of an $(m,\ell)$-cluster.

Every $(m,d)$-\emph{daisy} with $d=O(\log m)$ is a cluster, i.e., it satisfies Definition~\ref{def:cluster} with a logarithmic gap between the diameter and the logsize.
Indeed, from the definition of an $(m,d)$-daisy  it follows that for all $x_1,x_2$ in this set we have 
\[
\C(x_1 \cnd x_2) \le \C(z \cnd x_2) + \C(x_1 \cnd z) + O(\log m) .
\]
A similar bound applies to $\C(x_2 \cnd x_1)$.  Therefore,
\[
\dist(x_1, x_2) \le m + O(d+\log m) = m + O(\log m).
\]
The cardinality of this set is at least $2^m$ since it contains all pairs $\langle z,w\rangle$ where  $w$ is an $m$-bit string. 
Thus, this set is an $(m + O(\log m), m)$-cluster.

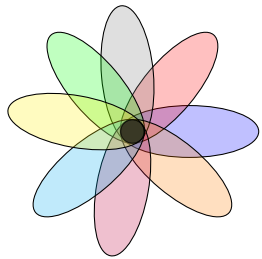
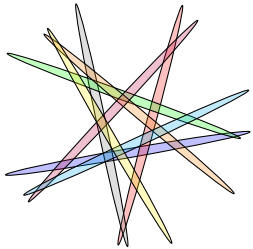
\begin{figure}
\centering
\begin{subfigure}{.45\textwidth}
\begin{center}
\begin{tikzpicture}[scale=0.23]

    \draw[xshift=3.3cm, yshift=0.0cm, rotate=0,  fill opacity=0.25, fill=blue] (0,0) ellipse (4cm and 1.5cm)  ;
    \draw[rotate=50, xshift=3.3cm, yshift=0.0cm,  fill opacity=0.25, fill=red] (0,0) ellipse (4cm and 1.5cm)  ;
    \draw[rotate=95, xshift=3.3cm, yshift=0.0cm,  fill opacity=0.25, fill=gray] (0,0) ellipse (4cm and 1.5cm)  ;
    \draw[rotate=130, xshift=3.3cm, yshift=0.0cm,  fill opacity=0.25, fill=green] (0,0) ellipse (4cm and 1.5cm)  ;
    \draw[rotate=-40, xshift=3.3cm, yshift=0.0cm,  fill opacity=0.25, fill=orange] (0,0) ellipse (4cm and 1.5cm)  ;      
    \draw[rotate=-100, xshift=3.3cm, yshift=0.0cm,  fill opacity=0.25, fill=purple] (0,0) ellipse (4cm and 1.5cm)  ;   
    \draw[rotate=-140, xshift=3.3cm, yshift=0.0cm,  fill opacity=0.25, fill=cyan] (0,0) ellipse (4cm and 1.5cm)  ;   
    \draw[rotate=170, xshift=3.3cm, yshift=0.0cm,  fill opacity=0.25, fill=yellow] (0,0) ellipse (4cm and 1.5cm)  ;   
    
       \draw[opacity=0.7, fill=black!95] (0,0) ellipse (0.68cm and 0.68cm);    
\end{tikzpicture}
\caption{A bunch of intersecting sets having a~common core.}\label{f:sunflower}
\end{center}
\end{subfigure}%
%\vspace{2pt}
\rule{4pt}{0pt} 
\begin{subfigure}{.45\textwidth}
\begin{center}
\begin{tikzpicture}[scale=0.075]

    \draw[xshift=0.3cm, yshift=-5.0cm, rotate=0,  fill opacity=0.25, fill=blue, rotate=10] (0,0) ellipse (22cm and 0.8cm)  ;
    \draw[xshift=4.3cm, yshift=0.0cm,  fill opacity=0.25, fill=red, rotate=75, ] (0,0) ellipse (22cm and 0.8cm)  ;
       \draw[xshift=-4.3cm, yshift=0.0cm,  fill opacity=0.25, fill=gray, rotate=102, ] (0,0) ellipse (22cm and 0.8cm)  ;
       \draw[xshift=-0.4cm, yshift=5.2cm,  fill opacity=0.25, fill=green,rotate=160, ] (0,0) ellipse (22cm and 0.8cm)  ;
       \draw[xshift=2.3cm, yshift=2.0cm,  fill opacity=0.25, fill=orange,rotate=-40,] (0,0) ellipse (22cm and 0.8cm)  ;      
       \draw[xshift=-2.3cm, yshift=2.5cm,  fill opacity=0.25, fill=purple, rotate=47,] (0,0) ellipse (22cm and 0.8cm)  ;   
       \draw[xshift=1.8cm, yshift=-3.1cm,  fill opacity=0.25, fill=cyan, rotate=-155,] (0,0) ellipse (22cm and 0.8cm)  ;   
    \draw[xshift=-2.9cm, yshift=-1.9cm,  fill opacity=0.25, fill=yellow, rotate=120,] (0,0) ellipse (22cm and 0.8cm)  ;

\end{tikzpicture}
\caption{A bunch of pairwise intersecting sets with no common core.}\label{f:nocore}
\end{center}
\end{subfigure}
\caption{Geometric representation of clusters.}
\end{figure}

A daisy is by definition a cluster with an explicitly given core.
A  daisy of strings intuitively resembles  a flower with a core and many petals. We can draw it on the plane as a family of $2^m$ pairwise intersecting sets (so that the mutual information of every two strings corresponds to the size of the intersection between two sets) such that all these sets have a common part (the core), as shown in Fig.~\ref{f:sunflower}, and the size of each ``petal'' (outside the core) is at most $m$. However,  a naive parallelism between the mutual information and intersections of sets  can be  deceiving. Indeed, G\'acs and K\"orner  showed that the mutual information of two strings may not correspond to any material object (see \cite{gacs-korner}). Moreover, even if the mutual information of each pair of objects can be ``materialized,'' it seems possible that pairwise intersecting objects  do not share any common core, as in the example in Fig.~\ref{f:nocore}.

So a natural question arises: do there exist clusters substantially different from a daisy? Rather surprisingly, it turns out that  there are no other  clusters besides the daisies and their subsets. We can say informally that all  clusters (in the sense of Definition~\ref{def:cluster}) resemble  Fig.~\ref{f:sunflower} and not  Fig.~\ref{f:nocore}.
This is the  main result of this paper.

\begin{theorem*}[informal version] 
Every  cluster in the sense of information distance is a sufficiently large subset of some daisy. 
\end{theorem*}

Observe that every large enough  %(with approximately $2^m$ elements)
subset of a cluster is still a cluster (of high enough density). Thus, this theorem may be interpreted as a description of all dense enough clusters as sufficiently large parts of daisies.

Now we proceed with a more formal statement.
\begin{theorem}[\textbf{main result}, the formal version]\label{th:main}
Let $S$ be a set of strings such that $\C(x \cnd x') \le m$ for every two strings $x,x'\in S$.
Assume that $\log \#S \ge m - d $ for some $d$. Then there exists a string  $z$ such that 
\[
\C(z \cnd x) \le  O(d + \log m) \text{ and } \C(x \cnd z) \le m+ O(d + \log m)
\]
for all $x\in S$.
\end{theorem}

Theorem~\ref{th:main} can be naturally rephrased in terms of clusters and daisies: 
it claims that for every $(m,m-d)$-cluster $S$ 
there exists a string $z$ such that $S$ is included in an
$(m + O(d+\log m),O(d+\log m))$-daisy 
with the core $z$. Notice that the found core $z$ possibly does not belong to $S$.
In fact, a cluster may even not contain  any element close to its core.

\begin{proof} 
We start the proof with the following lemma.

\begin{lemma}[Many paths $x$ -- $y$ -- $z$ imply one shorter path $x$ -- $z$]\label{lem:path}
Assume that for given strings $x$, $z$ and for given numbers $u,v,w$ there are at least $2^u$ strings $y$ such that 
\[
\C(y\cnd x) < v \quad \text{and}\quad \C(z\cnd y)<w.
\]
Then $\C(z\cnd x)\le v+w-u+O(\log(v+w))$.
\end{lemma}

\begin{proof}[Proof of Lemma~\ref{lem:path}]
Given $x$, $v$, $w$, we can enumerate all $z$ for which there exists $y$ with the required properties. There are at most $2^{v+w}$ paths of length $2$ and at least $2^u$ of these paths should lead to such a $z$. So there are at most $2^{v+w-u}$ different $z$, and this implies the  bound for $\C(z\cnd x)$.
\end{proof}

The previous lemma can be used to merge clusters, as the following remark shows.
\begin{remark*}%[Merging two clusters]
\label{lem:merge}
(Merging two clusters).
If $S$ and $S'$ are two clusters of diameter $m$ that have at least $2^{m-d}$ common elements, then  their union is a cluster  of diameter at most $m+d+O(\log m)$.
%\end{lemma}
%\begin{proof}[Proof of Lemma~\ref{lem:merge}] 
Indeed, if $x$ and $x'$ are elements from $S$ and $S'$ respectively, then there are at least $2^{m-d}$ paths $x-x''-x'$ such that $x''\in S\cap S'$. Therefore, from Lemma~\ref{lem:path} it follows that $\dist(x,x')\le 2m-(m-d)+O(\log m)=m+d+O(\log m)$.
\end{remark*}

More specifically, we will need  the following version of cluster merging:

\begin{lemma}[Merging a cluster with a daisy]\label{lem:merge-canonical}
Assume that an $(m+d_1,m-d_2)$-cluster $S$  has at least $2^{m-d_3}$ common elements with an $(m,d_4)$-daisy $S'$ with a core $z$. Then $S$ is contained in the daisy $S''$ with the same core $z$ with the parameters 
\begin{equation}\label{eq:can}
\left(  m + O\big({\textstyle\sum\limits_i} d_i+\log m\big),  O\big({\textstyle\sum\limits_i} d_i+ \log m\big) \right). 
\end{equation}
\end{lemma}

\begin{proof}
For every $x\in S$ there are at least $2^{m-d_3}$ chains $z-x'-x$ such that $x'\in S\cap S'$. From Lemma~\ref{lem:path} it follows that 
\[
\C(z \cnd x) \le m+d_1 + d_4 - (m-d_3) + O(\log m) =  O\big({\textstyle\sum\limits_i} d_i+ \log m\big)
\]
and
\[
\C(x \cnd z) \le m+d_4 + m+ d_1 - (m-d_3) + O(\log m) = m+ O\big({\textstyle\sum\limits_i} d_i+ \log m\big).
\]
Therefore, $x$ belongs to the daisy $S''$ with parameters \eqref{eq:can} and the base $z$.
\end{proof}
To prove the theorem, it is enough (thanks to Lemma~\ref{lem:merge-canonical}) to find a daisy with parameters
\[
\left( m + O\big( d+ \log m\big),  O\big( d +\log m\big) \right)
\]
that has a large (of cardinality at least $2^{m-O\big( d+\log m\big)}$) intersection with the given cluster~$S$. We do it as follows. The property of being a cluster with given parameters is enumerable. So we can run a process enumerating all $(m,m-d)$-clusters.
We do not restrict the length or complexity of strings in the clusters, so the enumeration will be infinite. As any other cluster with the same parameters,  our cluster $S$ will be enumerated at some stage of this process.

 Let us fix some threshold $d'$ (that will be slightly greater than $d$, see below). We say that two clusters $S_1, S_2$ have a \emph{large}  intersection if $\#(S_1\cap S_2) > 2^{m-d'}$.  To make the enumeration procedure defined above more economic, we will drop some clusters from this enumeration. We  will keep only the ones that do not have large intersections with one of the  clusters enumerated (and not dropped) earlier.  We call the clusters that are not dropped \emph{referential clusters}. We assign  to the referential clusters their ordinal numbers in the order they appear  in the enumeration. (Observe again that there can be infinitely many referential clusters.) We will see that either $S$ itself or some cluster that has a  large intersection with $S$ will become a referential cluster, and we plan to take its ordinal number in the enumeration as the core $z$  of the daisy  that we are looking for. 
 
First of all, we argue that every referential cluster is a part of  a daisy with parameters 
\[
\big( m + O(d + \log m), O(d + \log m)\big).
\]
Let $S_i$ be the referential cluster with ordinal number $i$.  
We start with the observation that every element $x\in S_i$ can be determined if we know $i$ and the ordinal number of $x$ in some standard ordering of $S_i$. 
To organize the process of enumeration of the referential clusters we  also need to know the number $m$, which requires $O(\log m)$ bits. Therefore, $\C(x\cnd i)\le m+O(d+\log m)$ for all $x\in S_i$. We want to show that $S_i$ is a part of a daisy with a core $i$. To this end, we show that $\C(i\cnd x)\approx 0$ for all $x\in S_i$. We use the following lemma saying that (under some conditions on the parameters) the multiplicity of the family $S_i$ is small, i.e., every $x$ is covered by only a small number of $S_i$. Then, to reconstruct $i$ given $x$, we need only the ordinal number of $S_i$ in the list of referential clusters containing $x$ (and also the number $m$, as before).

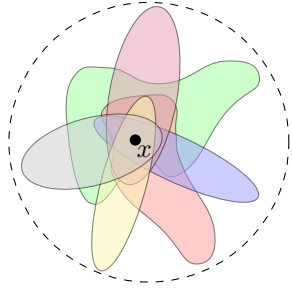
\begin{figure}
\centering
  \begin{tikzpicture}[scale=0.3]
            \draw[fill=green!40, opacity=0.5]  plot[smooth, tension=.9] coordinates {(-3.0,0.5) (-2.8,2.7) (-1.2,3.2) (1.2,2.5) (4,3.5) (5.5,2.7) (4.0,1.2) (2.0,-1.5) (0.2,-1.0) (-2,-2.8) (-3.0,0.5)};
           \draw[fill=red!40, opacity=0.5,rotate=-90]  plot[smooth, tension=.9] coordinates {(-2.0,1.2) (-1.5,1.8) (-0.3,1.8) (1.2,2.1) (4,3.5) (5.5,2.7) (4.0,1.2) (2.0,-0.8) (0.2,-1.0) (-1.5,-1.3) (-2.0,1.2)};

             \draw[fill=blue!40, opacity=0.5,rotate=-25] (2,0) ellipse (4 and 1);
             \draw[fill=purple!40, opacity=0.5,rotate=80] (2,0) ellipse (4 and 1.5);
             \draw[fill=yellow!40, opacity=0.5,rotate=-105] (2,0) ellipse (4 and 1);

             \draw[fill=gray!40, opacity=0.5,rotate=-165] (2,0) ellipse (3.2 and 1.5);             
            
          \draw[dashed] (0.6,-0.1) ellipse (6.2 and 6.2);

             \node at (0,0) [circle,fill,inner sep=1.5pt]{};
             \node at (0.4,-0.5){$x$};
                         
        \end{tikzpicture}
\caption{An element $x$ (shown as a dot) is covered by a family of sets (shown in different colors) all of which lie in a neighborhood of this element (shown as an area with a  dashed borderline).}    
\label{f:cover}    
\end{figure}

\begin{lemma}[Multiplicity bound]\label{lem:multiplicity}\label{l:multiplicity}
Assume that $d'>2d+1$. Then every string $x$ can be covered by at most $2^{d+1}$ referential $(m,m-d)$-clusters.
\end{lemma}

\begin{proof}[Proof of Lemma~\ref{lem:multiplicity}]
Assume that some string $x$ is covered by $N$ referential clusters. Each cluster $S_i$ has size at least $2^{m-d}$, and the intersection of every two clusters  is at most $2^{m-d'}$ (otherwise the second cluster could not be  selected as a referential one).  Note also that all elements of all clusters that contain $x$ have conditional complexity at most $m$  conditional on $x$. Therefore, the union of all these clusters  has a size at most $2^{m+1}$  (all clusters covering $x$ are included in a rather small neighborhood of $x$, see Fig.~\ref{f:cover}).

We use the following probabilistic claim:
\begin{claim} Let $\eps$ be the inverse of a positive integer number\footnote{A similar claim is true for all real numbers $\eps>0$. The assumption that $1/\eps$ is an integer number slightly simplifies the calculations since we can ignore rounding.}.
 If there are $N$ events of probability greater than $\eps$, and all pairwise intersections have probability less than $\eps^2/2$, then $N < 2/\eps$.
 \end{claim}\label{claim}
 \noindent
 (The bounds in the claim are pretty tight: $1/\eps$ events of probability $\eps$ could be disjoint, and any number of independent events of probability $\eps$ have intersection~$\eps^2$.)
 \begin{proof}
Assume that we have $N=2/\eps$ events (we decrease $N$ if needed). 
By  the principle of inclusion and exclusion, 
 the probability of the union of these events is strictly greater than
\[
N\cdot \eps - \frac{N^2}{2}\cdot\frac{\eps^2}{2} =2- \frac{4}{2\eps^2}\cdot \frac{\eps^2}{2} = 1,
\]
a contradiction. 
\end{proof}

We consider the union of all $(m, m -d)$-clusters covering $x$  as the probability space with equiprobable points, where each cluster is an event.
To apply this claim, we note that each cluster  has a probability of at least $\eps:=2^{-d}$, and the intersections are of probability at most $2^{-d'}$. Since $d'>2d+1$, we can apply the Claim  and obtain the bound $2^{d+1}$ for the number of clusters covering $x$. Therefore, we have
\[ 
\C(x\cnd i)\le m+O(\log m) \ \ \text{and}\ \ \C(i\cnd x)\le d+O(\log m)
\]
for every element $x$ of every referential cluster $S_i$. Thus, each referential cluster is a part of an $\big(m + O(d+\log m), O(d+\log m) \big)$-daisy.
\end{proof}% of lem:multiplicity

Now we bind together all parts of the  argument. Assume that we have an $(m-d,m)$-cluster $S$. We let $d'=2d+2$ (so the  condition of Lemma~\ref{l:multiplicity} is true) and start the process of enumeration of referential clusters using $d'$ as the ``large intersection'' threshold. The construction guarantees that $S$ is one of the referential clusters or at least it has a large intersection with some referential cluster $S_i$. Every referential cluster  $S_i$ is a part of an  
$\big(m + O(d+\log m),O(d+\log m)\big)$-daisy. Therefore, we can apply Lemma~\ref{lem:merge-canonical} and conclude that $S$ is a part of a slightly bigger 
$\big(  m + O(d+\log m), O(d+\log m) \big)$-daisy,
and the theorem is proven. 
\end{proof}% of main theorem

\section{Clusters and the mutual information of a triple}
In this section, we discuss an application of Theorem~\ref{th:main} that motivated the definition of ``bunches'' proposed  in \cite{romashchenko-ccc-2003} (similar to the definition of clusters discussed in the  previous section).
\begin{theorem}[\cite{romashchenko-ccc-2003,romashchenko-ppi-2003}]
\label{th:extraction}
For every triple of strings $x,y,z$ there exists a string $w$ such that
\[
\C(w) = \I(x:y:z) + O(\eps+\log \C(x,y,z))
\]
and
\[
\max\{ \C(w \cnd x), \C(w \cnd y), \C(w \cnd z) \} = O(\eps +\log \C(x,y,z)),
\]
where $\eps := \max\big\{ \I(x:y \cnd z) , \I(x:z \cnd y) , \I(y:z \cnd x) \big\}$ and
\[
\I(x:y:z) := \C(x) + \C(y) + \C(z) - \C(x,y) - \C(x,z) - \C(y,z) + \C(x,y,z).
\]
\end{theorem}
\noindent
In particular, if the three values of conditional mutual information $\I(x:y \cnd z)$, $\I(x:z \cnd y)$, $\I(y:z \cnd x)$ are negligibly small (say, logarithmic in $\C(x,y,z)$) as shown in  Fig.~\ref{fig:venn}, then the mutual information shared by $x,y,z$ can be  materialized in the sense of \emph{common information} by G\'acs and K\"orner, \cite{gacs-korner}.
\def\firstcircle{(0,0) circle (1.5cm)}
\def\secondcircle{(0:2cm) circle (1.5cm)}
\def\thirdcircle{(-60:2cm) circle (1.5cm)}
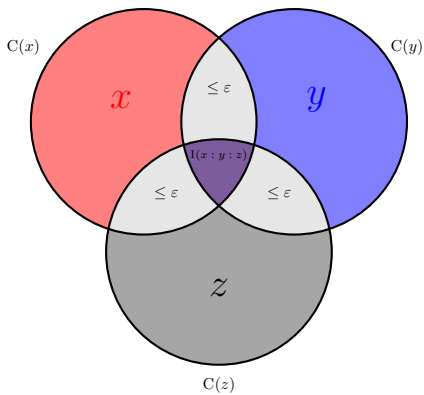
\begin{figure}[h]
\begin{center}
\begin{tikzpicture}[scale=1.0, every node/.style={scale=0.6}]

    \begin{scope}[fill opacity=0.5]
        \fill[red] \firstcircle;
        \fill[blue] \secondcircle;
        \fill[black!70] \thirdcircle;
    \end{scope}

    \begin{scope}[fill opacity=1]
    	\clip \firstcircle;
        \fill[black!10] \secondcircle;
        \fill[black!10] \thirdcircle;
    \end{scope}
    \begin{scope}[fill opacity=1]
    	\clip \secondcircle;
        \fill[black!10] \thirdcircle;
    \end{scope}
    \begin{scope}[fill opacity=0.5]
     	\clip \firstcircle;   
    	\clip \secondcircle;	
    	\clip \thirdcircle;	
        \fill[red] \firstcircle;
        \fill[blue] \secondcircle;
        \fill[gray] \thirdcircle;
    \end{scope}

        \draw[thick] \firstcircle node[below] {};
        \draw[thick] \secondcircle node [below] {};
	\draw[thick] \thirdcircle node [below] {};
	
	    \draw  (-1.6,1.0) node  {$\C(x)$};  	    \draw  (-0.3,0.3) node  {\color{red}\Huge $x$};
        \draw  (3.5,1.0) node  {$\C(y)$};  \draw  (2.3,0.3) node  {\color{blue}\Huge $y$};
         \draw  (1.0,-3.5) node {$\C(z)$}; \draw  (1.0,-2.2) node  {\color{black}\Huge $z$};
         \draw  (1.0,-0.45) node  {\scriptsize{$\I(x: y: z)$}};
         \draw  (1.0,-2.5) node  {};
         \draw[ opacity=1.0 ]  (1.0,0.45) node  {{$\le \eps$}};
         \draw[ opacity=1.0 ]  (1.8,-0.95) node  {{$\le \eps$}};
         \draw[ opacity=1.0 ]  (0.3,-0.95) node  {{$\le \eps$}};
\end{tikzpicture}
\caption[]{A Venn-like diagram representing information quantities for a triplet $(x,y,z)$. The area of each circle represents the value of Kolmogorov complexity of $x$, $y$, and $z$ respectively. The areas of the unions of any two circles represent Kolmogorov complexity of  pairs, and the area of the union of all three circles represents Kolmogorov complexity of the triple.  Accordingly, the intersections of every two circles represent the mutual information of pairs, and the intersection of all three circles represents the value $\I(x:y:z)$. The areas shown in light gray represent the values of the three conditional mutual information $\I(x:y \cnd z)$, $\I(x:z \cnd y)$, $\I(y:z \cnd x)$. 
Theorem~\ref{th:extraction} claims that if the three values of conditional mutual information are negligibly small, then the mutual information of the triple $\I(x:y:z)$ (which is  in this case $\eps$-close to each of the values $\I(x:y)$, $\I(x:z)$, and $\I(y:z)$) can be materialized.}\label{fig:venn}
\end{center}
\end{figure}

\begin{proof}
For each triple of strings $(x,y,z)$ we call by its \emph{complexity profile} the tuple of seven complexity quantities
\[
\big(
\C(x), \C(y), \C(z), \C(x,y), \C(x,z), \C(y,z), \C(x,y,z)  
\big).
\]
For every precision parameter (an integer number) $\delta$ we define the set ${\cal C}_\delta $ of $\delta$-\emph{clones of $z$ conditional on} $(x,y)$ as the set of all $z'$ such that
the complexity profile of $(x,y,z')$ differs in each component from the complexity profile of $(x,y,z)$ by at most $\delta$. 

A simple counting argument implies (see, e.g., \cite{clones})
that there exists a constant $D$ (independent of $x,y,z$) such that for all $x,y,z$ the set of $\delta$-clones of $z$ conditional on $(x,y)$ with $\delta = D\log \C(x,y,z)$ consists of 
$
2^{\C(z \cnd x,y) - O(\delta)}
$
strings $z'$. In what follows we fix such a $\delta$.

We claim that  ${\cal C}_\delta $ is a  cluster. To prove this fact we need the following inequality, where $z'$ and $z''$ are two arbitrary elements from ${\cal C}_\delta $:
\[
\begin{array}{rcl}
\I(x:y) &\le& \I(x:y \cnd z') + \I(x:y \cnd z'') + \I(z':z'') \\
&&{} + \I(x:y \cnd z) + \I(x:z \cnd y) + \I(y:z \cnd x) + O(\log \C(x,y,z)).
\end{array}
\]
This inequality\footnote{This inequality is a \emph{non--Shannon type} one, i.e., it cannot be represented as a linear combination of several instances of  inequalities representing non negativity of conditional Kolmogorov complexity, or mutual information, or conditional mutual information.
The very first example of a non--Shannon type linear inequality for Shannon's entropy was proven by Zhang and Yeung in \cite{zhang-yeung}. It is known, see \cite{hammer}, that  the same linear inequalities are true for Shannon's entropy and Kolmogorov complexity. The inequality used in our proof is a little generalization of the inequality discovered by Zhang and Yeung, see \cite{mmrv} for details.} is valid for all strings $x,y,z,z',z''$, see \cite{mmrv}.
Since $z'$ and $z''$ are clones of $z$ conditional on $x,y$, the inequality rewrites to
\[
\begin{array}{rcl}
\I(x:y) &\le& \I(z':z'') + O(\eps+\log \C(x,y,z)).
\end{array}
\]
It follows that
\[
\begin{array}{rcl}
\C(z'' \cnd  z') &=&   \C(z'') - \I(z':z'') \\
&\le & \C(z)   - \I(x:y) + O(\eps+\log \C(x,y,z))\\
&=& \C(z  \cnd x,y)  + O(\eps+\log \C(x,y,z)).
\end{array}
\]
Therefore,  ${\cal C}_\delta $ is a cluster with the parameters 
\[
\big(\C(z \cnd x,y) + O(\delta+\log \C(x,y,z)), \C(z \cnd x,y) - O(\delta+\log \C(x,y,z)) \big).
\]
Theorem~\ref{th:main} implies that there exists a string $w$ (the core of the cluster) such that 
$\C(w \cnd \hat z) = O(\eps +\log \C(x,y,z))$ and $\C(\hat z \cnd w) = \C(z \cnd x,y) + O(\eps +\log \C(x,y,z))$ for all $\hat z \in {\cal C}_\delta$ (including the original string $z$).  
It is easy to compute Kolmogorov complexity of  $w$:
\[ \C(w) = \I(z:\langle x,y\rangle ) + O(\eps + \log \C(x,y,z)) =  \I(x:y:z) + O(\eps + \log \C(x,y,z)).\] 
It remains to observe that due to Lemma~\ref{lem:path}
\[
\C(w \cnd x) = O(\eps+\log \C(x,y,z))
\text{ and }
\C(w \cnd y) = O(\eps+\log \C(x,y,z))
\]
 (it is enough to count the number of chains $x$ -- $\hat z$ -- $w$ and $y$ -- $\hat z$ -- $w$ with $\hat z\in {\cal C}_\delta$).
\end{proof}
 
\section{Discussion}

The questions addressed in this paper seem to be related to the density properties studied in \cite[Section~IX]{InformationDistance1998}, where the authors estimated the rate of growth of the number of elements in balls of radius $r$ in the metric spaces induced by the information distance. We should stress, however, that a ball (the set of  strings $x'$ at the distance at most $r$ from a given center $x$) is not a cluster in the sense of Definition~\ref{def:cluster}. 

The main result of this theorem is formulated and proven with a ``logarithmic precision,'' which is quite typical for the theory of Kolmogorov complexity.  In many applications the logarithmic precision is enough.
At the same time, it seems that the residue terms in Theorem~\ref{th:main} can be made somewhat tighter. 
In this vein, an anonymous referee of the \emph{Theoretical Computer Science} journal
suggested the following stronger version of Theorem~\ref{th:main}:
\begin{theorem-bis*}
Let $S$ be a set of strings such that $C(x|x') \le m$ for every two strings $x,x'\in S$. Assume that $\log \#S \ge m - d$ for some $d$. Then there exists a string $z$ such that $C(z|x, m) < O(d)$ and $C(x|z, m) < m + O(d)$.
\end{theorem-bis*}
\noindent
(Compared with Theorem~\ref{th:main}, the conclusion of this statement contains no additive terms $O(\log m)$; on the other hand,  the number  $m$ is included into the conditions of two expressions with  Kolmogorov complexity.) This version of the theorem  can be proven by an argument very similar to the proof of Theorem~\ref{th:main} presented in Section~\ref{s:2}; we only need to relativize to  $m$ the expressions with Kolmogorov complexity that appear in the proof. 
However, if we want to rephrase  Theorem~${\ref{th:main}'}$ in terms of clusters and daisies (cf. the paragraph after Theorem~\ref{th:main} on p.~\pageref{th:main}), we would need to revise the definition of a daisy.
This observation suggests that the definitions of clusters and daisies might need to be refined.

Let us mention also that  slightly different  variants of the definition of a cluster may be helpful in some applications, see \cite{muchnik-romashchenko-ppi}.
An interesting variant of the definition  was proposed by S.~Epstein in \cite{sam}, where the principal parameter was not the \emph{maximum} but the \emph{average} distance between elements of a cluster.  Epstein argued that the {density} of a cluster is connected with  the mutual information between this cluster  and the \emph{halting sequence} (characterizing the stopping Turing machine in the universal enumeration).
Thus, the formulation of the most natural and practical definition of a dense cluster (in the sense of information distance) remains an open question.

\paragraph{Acknowledgments.} The author is grateful to  Alexander Shen and Marius Zimand for fruitful discussions, 
especially for the elegant form of the probabilistic claim 
(see the \emph{Claim} on p.~\pageref{claim})
suggested by Alexander Shen.
The author also thanks the anonymous referees of the Theoretical Computer Science journal for the careful review of the paper and valuable comments and suggestions.

\bibliographystyle{unsrt}

\bibliography{clusters}
\end{document}